\newif\iflong
\longtrue

\documentclass[a4paper,USenglish]{lipics-v2018}

\usepackage{microtype}

\usepackage{amsmath, amsthm, amssymb}
\usepackage{mathtools}
\usepackage{cite}
\usepackage{url}
\usepackage{appendix}
\usepackage{graphicx}
\usepackage{epstopdf}
\usepackage{setspace}
\usepackage{enumerate}
\usepackage{color}
\usepackage{xcolor}
 \usepackage[noend]{algorithmic}
 \usepackage{algorithm}
\usepackage{multirow}
\usepackage{paralist}

\usepackage{caption}
\usepackage{xspace}
\usepackage{tikz}

\newcommand{\mypar}[1]{\medskip\noindent\textbf{#1. }}
\newcommand{\alg}[1]{\textsc{#1}\xspace}

\theoremstyle{plain}
\newtheorem{observation}[theorem]{Observation}

\newcommand{\hide}[1]{}

\newcommand{\whp}[1][\empty]{\ensuremath{\text{w.h.p.}\ifthenelse{\equal{#1}{\empty}}{}{(#1)}}}
\newcommand{\Whp}[1][\empty]{\ensuremath{\text{W.h.p.}\ifthenelse{\equal{#1}{\empty}}{}{(#1)}}}

\renewcommand{\Pr}{\mathbb{P}}

\DeclareMathOperator{\E}{\mathbb{E}}

\DeclareMathOperator*{\argmax}{arg\,max}

\newcommand{\FullOrShort}{short}

\ifthenelse{\equal{\FullOrShort}{full}}{
	  
  \newcommand{\fullOnly}[1]{#1}
  \newcommand{\shortOnly}[1]{}

  }{
    \usepackage{times}

    \newcommand{\fullOnly}[1]{}
    \newcommand{\shortOnly}[1]{#1}
  }

\newcounter{tempAppCounter}
\newcounter{tempSecNumber}
\newcommand{\setAppCounter}[1]{
  \setcounter{tempAppCounter}{\arabic{theorem}}
  \setcounter{tempSecNumber}{\arabic{section}}

  \setcounter{theorem}{\arabic{ctr:#1}}
  \setcounter{section}{\arabic{ctrsec:#1}}
  \renewcommand{\thetheorem}{\arabic{section}.\arabic{theorem}}
}

\newcommand{\renewAppCounter}{
  \setcounter{theorem}{\arabic{tempAppCounter}}
  \setcounter{section}{\arabic{tempSecNumber}}
  \def\thetheorem{\oldtheorem}
}

\newcounter{note}[section]

\title{Distributed Algorithms for Minimum Degree Spanning Trees}

\author{Michael Dinitz}{Department of Computer Science, Johns Hopkins University, Baltimore, MD, USA.}{mdinitz@cs.jhu.edu}{}{Supported in part by NSF awards 1464239 and 1535887.}

\author{Magn\'us M. Halld\'orsson}{ICE-TCS, School of Computer Science, Reykjavik University, Iceland}{mmh@ru.is}{https://orcid.org/0000-0002-5774-8437}{Supported by grants nos.~152679-05 and 174484-05 from the Icelandic Research Fund. 
Also supported by the Research Institute for Mathematical Sciences, a Joint
Usage/Research Center located in Kyoto University.}

\author{Calvin Newport}{Department of Computer Science, Georgetown University, Washington, DC, USA}{cnewport@cs.georgetown.edu}{}{Supported in part by NSF awards 1733842 and 1649484.}

\authorrunning{M. Dinitz, M.\,M. Halld\'orsson and C. Newport}

\Copyright{Michael Dinitz and Magn\'us M. Halld\'orsson and Calvin Newport}

\subjclass{Theory of computation $\rightarrow$ Distributed algorithms}


\EventEditors{John Q. Open and Joan R. Access}
\EventNoEds{2}
\EventLongTitle{32nd International Symposium on Distributed Computing (DISC 2018)}
\EventShortTitle{DISC 2018}
\EventAcronym{DISC}
\EventYear{2018}
\EventDate{October 15--19, 2018}
\EventLocation{New Orleans, Lousiana}
\EventLogo{}
\SeriesVolume{42}
\ArticleNo{23}
\iflong
\hideLIPIcs  
\nolinenumbers
\fi

\begin{document}

\maketitle

\begin{abstract}
The {\em minimum degree spanning tree} (MDST) problem
 requires the construction of a spanning tree $T$ for graph $G=(V,E)$ with $n$ vertices, such that the maximum degree $d$ of $T$
 is the smallest among all spanning trees of $G$.
 In this paper, we present two new distributed approximation algorithms for the MDST problem.
 Our first result is a randomized distributed algorithm that constructs a spanning tree of maximum degree $\hat d = O(d\log{n})$.
 It requires $O((D + \sqrt{n}) \log^2 n)$ rounds (w.h.p.), where $D$ is the graph diameter,
 which matches (within log factors) the optimal round complexity for the related minimum spanning tree problem.
 Our second result refines this approximation factor by constructing a tree with maximum degree $\hat d = O(d + \log{n})$,
 though at the cost of additional polylogarithmic factors in the round complexity. 
 Although efficient approximation algorithms for the MDST problem have been known in the sequential setting since the 1990's,
 our results are first efficient distributed solutions for this problem. 
 \end{abstract}


\section{Introduction \& Related Work}
\label{sec:intro}

We present two new distributed approximation algorithms for the {\em minimum degree spanning tree} (MDST) problem,
which requires the construction of a spanning tree $T$ for graph $G=(V,E)$ with $n$ vertices, such that the maximum degree of $T$
 is the smallest among all spanning trees of $G$.
 As argued in~\cite{fr92,fr94}, in addition to their theoretical interestingness,
 these trees are particularly useful in network communication scenarios in which 
 low-degree backbones reduce routing overhead.
 
In the {\em sequential} setting, the problem is easily shown to be NP-hard (by reduction from the Hamiltonian path problem).
The best known approximation is due to F\"{u}rer and Raghavachari~\cite{fr94},
who provide a polynomial-time algorithm that constructs a tree with maximum degree $d+1$,
where $d$ is the minimum maximum degree over all spanning trees in the graph. 
To the best of our knowledge, 
there exist no efficient {\em distributed} approximation algorithm for the MDST problem. 

This paper addresses this gap.
In more detail, we present two new distributed approximation algorithms for the MDST problem.
Our first algorithm guarantees a spanning tree with a maximum degree in $O(d\log{n})$
and a round complexity that is comparable to the optimal solutions to the related {\em minimum spanning tree} problem.
Our second algorithm guarantees a maximum degree in $O(d+\log{n})$, but at the cost of 
extra polylogarithmic factors in the round complexity.


 {\bf Model.}
The results discussed in this paper (both our own and previous work) assume the CONGEST model of distributed computation.
In this model, 
the network is described as an $n$-node graph $G=(V,E)$, with a computational process assigned to each node
and the edges representing communication channels.
Time proceeds in synchronous rounds.
In each round, each node can send a $O(\log{n})$-bit message to each of its neighbors in the graph.
Our new results actually work in the harder {\em broadcast} variation of the CONGEST model (broadcast-CONGEST)
in which nodes
must broadcast the {\em same} message to all of their neighbors in a given round.


 {\bf Background.}
The construction of spanning trees with useful properties is one the primary topics in the study of distributed graph algorithms.
The most well-studied problem in this area is the {\em minimum spanning tree} (MST) problem,
which requires the construction of a spanning tree that minimizes the sum of edge weights.
We will briefly summarize the relevant related work on the MST problem,
as it provides a template for the progression of results on the MDST problem studied in this paper.

In the 1980's, Gallager, Humblet, and Spira~\cite{28} help instigate this area
with a distributed algorithm that constructs an MST in $O(n\log{n})$ rounds (similar
ideas appeared in a 1926 paper by Boruvka~\cite{58} that was not translated into English until more recently).
A series of follow up papers~\cite{15,27,7} improved this complexity to $O(n)$ rounds, 
which is worst-case optimal in the sense that $\Omega(n)$ rounds are required in certain graphs with diameter $D=\Theta(n)$.

Garay, Kutten and Peleg~\cite{29}  isolated the graph diameter $D$ as a distinct parameter, enabling further progress.
They described a distributed MST algorithm that solves the problem in $O(D+n^{0.61})$ rounds,
which is sub-linear for graphs with sub-linear diameters.
This result was subsequently improved to $O(D+ \sqrt{n}\log^{*}{n})$ rounds~\cite{44} .
A series of lower bound results~\cite{64,23,17} established that any non-trivial approximation of an 
MST requires  $\Omega(D + \sqrt{n/\log{n}})$ rounds, even in graphs with small diameters.

To date, the MDST problem has been primarily studied in the context of sequential algorithms. 
In 1990, F\"{u}rer and Raghavachari~\cite{fr90} describe a polynomial time algorithm that constructs a tree with a maximum degree in $O(d\log{n})$
(recall that $d$ is the maximum degree of the optimal tree).\footnote{The result in~\cite{fr90} actually proves that finding a $\log{n}$-approximation is
in $NC$. All such solutions, however, can be simulated in polynomial time by a sequential algorithm, yielding the claimed polynomial-time result.} 
Agrawal, Klein and Ravi~\cite{akr} subsequently generalized this result to the Steiner tree variation
of the MDST problem.
F\"{u}rer and Raghavachari improved both results by presenting algorithms that guarantee a maximum degree  of $d+1$ for both the standard~\cite{fr92} and Steiner
tree~\cite{fr94} versions of the problem. Given that finding a spanning tree with maximum degree exactly $d$ is NP-hard,
these approximations are likely the best possible that can be achieved in polynomial time.

To the best of our knowledge, the first connection of the MDST problem to the distributed setting was made by Blin and Butelle~\cite{BB03},
who observed that the general strategy from~\cite{fr92} translates easily to distributed models.
This holds because the main mechanism in the sequential algorithm from~\cite{fr92} is a series of iterative improvements to an initial spanning tree,
in which each improvement reduces the degree of some high-degree node. 
Blin and Butelle note that each iteration of this search can be implemented with a small number of distributed broadcast and convergecasts
in distributed models with restricted message size (e.g., such as  CONGEST).

As with the original distributed solutions to the MST problem,
the distributed variation of~\cite{fr92} proposed in~\cite{BB03} requires $\Omega(n)$ rounds.\footnote{Isolating a specific round complexity claim from~\cite{BB03}
is complicated by the fact that they consider a different, pseudo-asynchronous model. Roughly speaking, however, directly implementing~\cite{fr92}
in the CONGEST model, using a BFS tree to implement the search, requires $O((\Delta-d)nD)$ rounds in the worst case, where $\Delta$ is the maximum degree of the
original graph.}
A key open question is whether  distributed solutions to the MDST problem
can follow the general trajectory of the MST, and
refine their efficiency to something closer to $\tilde{O}(D + \sqrt{n})$---which we suspect (though do not prove)
to be a lower bound for the MDST problem.

 {\bf Result \#1: Logarithmic Approximation.}
Our first algorithm constructs a spanning tree with maximum degree $\hat d = O(d\log{n})$ in $\tilde{O}(D + \sqrt{n})$ rounds,
with high probability in $n$ (w.h.p.).
This round complexity matches (within log factors) optimal solutions to the related MST problem.

Whereas the sequential algorithm from~\cite{fr92} begins with an arbitrary tree,
and then iteratively reduces its maximum degree, our distributed strategy begins with a forest of small trees,
and then carefully merges them in such a way that no individual node's degree grows too large.

In more detail, the algorithm proceeds in phases.
The input to each phase is a forest that covers the entire graph.
The goal of the phase is to combine enough of the trees in the forest to reduce their number by a constant factor,
while adding no more than $d$ new adjacent edges to any individual node.
These guarantees result in a single spanning tree after at most $p=O(\log{n})$ phases,
with the maximum degree of any individual node bounded by $p\cdot d = O(d\log{n})$.

Each phase proceeds in two steps.
During the first step, the algorithm computes a distributed maximal matching over the component graph defined by the forest.
If two components $C_i$ and $C_j$ are matched, they combine into one larger component.
This might not result in enough components merging, though, so in order to make more progress in the second step we consider the bipartite graph with left nodes corresponding to unmatched components and right nodes corresponding to low-degree vertices.  We prove that the existence of a degree-$d$ spanning tree implies that this bipartite graph must contain a subgraph in which the left nodes all have degree $1$ and the right node all have degree at most $d$, i.e., a $(1,d)$-matching.  So we find a maximal $(1,d)$-matching, which by standard arguments has size at least $1/2$ of the maximum, and so which results in a set of component merges including at least half of the remaining components.  

Since in every phase a constant fraction of the components are involved in a merge, there can be only $O(\log n)$ phase.  And in every phase, every node has its degree increased by at most $d+1$ (step 1 increases degrees by at most $1$, while step 2 increases degrees by at most $d$).  This gives the desired $O(\log n)$-approximation.

Implementing the above graph logic with efficient distributed primitives in the broadcast variation of the CONGEST model provides its own challenges.
For example: generalizing distributed matching strategies to execute over graph components (instead of single nodes),
and implementing intra-component communication without excess latency or congestion (a task which requires the treatment of small and large components to differ.)
Through careful optimization we are able to implement each of our $O(\log{n})$ phases in at most ${O}((D+\sqrt{n})\log{n})$ rounds.


 {\bf Result \#2: Refined Approximation.}
As obvious place to seek improvement on our first algorithm is in the magnitude of its approximation factor.
Whereas this algorithm constructs a spanning tree with maximum degree $\hat d = O(d\log{n})$,
the best known sequential algorithm achieves $\hat d = d+1$. 
Our second result aims to reduce this gap.
We present an algorithm that constructs a spanning tree with maximum degree $\hat d = O(d + \log{n})$.
To achieve this factor, however, requires a larger polylogarithmic factor in the round complexity
and a substantially more involved algorithm.

At a high level, the basic idea of this second algorithm is to attempt to parallelize a large number of the style of iterative 
improvements used in the original sequential solutions~\cite{fr92}.
Whereas the sequential algorithm improves the tree one edge at a time,
our second result enables many nodes to make large improvements to their degrees in a short period of time.  Since this algorithm is essentially a local search algorithm, its running time depends on the quality of the initial solution, and by using the output of our first algorithm as input to his algorithm we are able to save a logarithmic factor in the running time.


\section{Logarithmic Approximation}
\label{bcongest}

We describe and analyze an algorithm called \alg{MatchingMDST} (as in: {\em matching-based minimum degree spanning tree}).
Our goal is to prove the following:

\begin{theorem} \label{thm:bcast-aMDST}
With high probability in $n$: 
\alg{MatchingMDST} produces a spanning tree $T$ with maximum degree $\hat d = O(d\log{n})$
in $O\left((D + \sqrt{n}) \log^2{n}\right)$ rounds,
when executed in the broadcast-CONGEST model in a connected network graph of size $n>0$ and 
diameter $D$ that contains a spanning tree with maximum degree $d$.
\end{theorem}

To clarify the core ideas of the \alg{MatchingMDST} algorithm, we divide the description into four parts.
We begin in Section~\ref{sec:logn:1} by defining the types of matchings our algorithm uses to iteratively create our spanning tree.
Then in Section~\ref{sec:logn:2}, we define \alg{MatchingMDST} and analyze its correctness under the assumption
that its subroutines function correctly.
In Section~\ref{sec:logn:3}, 
we describe and analyze the low-level primitives used by \alg{MatchingMDST} (and the matching subroutines it calls) to efficiently disseminate information
within the components maintained by our algorithm.
Finally, in Section~\ref{sec:logn:4}, we describe and analyze the matching subroutines themselves.

\subsection{Matching Preliminaries}
\label{sec:logn:1}

For the following definitions and lemmas, we fix a graph $G = (V, E)$ with diameter $D$ with $n = |V|$.

\begin{definition} \label{def:component-matching}
Let $\mathcal P = \{C_1, C_2, \dots, C_k\}$ such that $C_i \subseteq V$ for all $i \in [k]$ and $C_i \cap C_j = \emptyset$ for all $i,j \in [k]$ with $i \neq j$.  Then $E' \subseteq E$ is a \emph{component matching} of $\mathcal P$ if for every $\{u,v\} \in E'$ there exists $i,j \in [k]$ with $i \neq j$ such that $u \in C_i$ and $v \in C_j$, and for every $i \in [k]$ there is at most one edge in $E'$ with an endpoint in $C_i$.
\end{definition}

Intuitively, a component matching is just a matching in the ``component graph" which has a vertex for each $C_i \in \mathcal P$ and an edge between $C_i$ and $C_j$ if there is an edge between the two components in $G$.  
In order to speed up our algorithm, we will want to also generalize this concept to $d$-matchings (where every node can have degree up to $d$), but restricted to a particular bipartite structure that will prove useful to our analysis.

\begin{definition} \label{def:component d-matching}
Let $U = \{C_1, C_2, \dots, C_k\}$ be a collection of disjoint sets of vertices.  Let $Q \subseteq V$ be a collection of vertices.
A \emph{$(1,d)$-component matching} of $(U, Q)$ is a collection of edges $E' \subseteq E$ such that every edge in $E'$ has one endpoint in $Q$ and the other endpoint in some $C_i \in U$, every vertex in $Q$ is incident on at most $d$ edges of $E'$, and for every $C_i \in U$ there is at most one edge of $E'$ with an endpoint in $C_i$.
\end{definition}

Similar to component matchings, the intuition behind a $(1,d)$-component matching is that if we look at the bipartite graph which has one vertex for each $C_i \in U$ on the left side and the vertices of $Q$ on the right side, with $C_i \in U$ adjacent to $v \in Q$ if there is an edge $\{u,v\} \in E$ with $u \in C_i$, then we are looking for a subgraph in which every left vertex (component in $U$) has degree at most $1$ and every right vertex (vertex in $Q$) has degree at most $d$.

A useful property of traditional matchings is that any {\em maximal} matching is at most a factor of $2$ smaller than the {\em maximum} matching on the same graph.
It is straightforward to prove that this same property holds for both component matchings and $(1,d)$-component matchings.

\begin{lemma} \label{lem:maximal-approx}
Any maximal component matching has size at least $1/2$ the size of any component matching, and any maximal $(1,d)$-component matching has size at least $1/2$ the size of any $(1,d)$-component matching.
\end{lemma}
\iflong
\begin{proof}
Let $M$ be a maximal component matching of $\mathcal P$, and let $M^*$ be an arbitrary component matching of $\mathcal P$.  Consider some edge $e \in M^* \setminus M$.  It cannot be added to $M$, so at least one of its endpoints is in a component which already has an incident edge in $M$.  Charge $e$ to this component (if both endpoints are in such components, choose one arbitrarily).  Then since $M^*$ is a component matching, every component in $\mathcal P$ gets charged at most its degree in $M \setminus M^*$.  Thus $|M^* \setminus M| \leq 2|M \setminus M^*|$, and so $|M^*| = |M^* \cap M| + |M^* \setminus M| \leq |M^* \cap M| + 2 |M \setminus M^*| \leq 2|M^*|$.

Similarly, let $M$ be a maximal $(1,d)$-component matching of $(U, Q)$, and let $M^*$ be an arbitrary $(1,d)$-component matching of $(U, Q)$.  Consider some $e \in M^* \setminus M$.  It cannot be added to $M$, so either its $U$ endpoint is in a component which already has an edge in $M \setminus M^*$ or its $Q$ endpoint has degree $d$ in $M$ (or both).  In the former case we charge this edge to the component containing its $U$ endpoint, and in the latter case we charge it to its $Q$ endpoint.  Clearly no component or vertex gets charged more than its degree in $M \setminus M^*$, and hence we know that $|M^* \setminus M| \leq 2|M \setminus M^*|$.  Thus $|M^*| \leq |M^* \cap M| + |M^* \setminus M| \leq |M^* \cap M| + 2|M \setminus M^*| \leq 2|M|$.   
\end{proof}
\fi

\subsection{The \alg{MatchingMDST} Algorithm}
\label{sec:logn:2}

We now present and analyze our main algorithm executed on a connected network $G=(V,E)$.
In the following, we assume that nodes know the optimal value $d$ (the minimum maximum degree over all spanning trees in the graph).
Below, we will show this assumption holds without loss of generality.

We call each iteration $i$ of the main {\bf for} loop  {\em phase $i$} of the algorithm.
During each phase $i$,
the \alg{MatchingMDST} algorithm calls three subroutines:
\alg{Component-Matching}($\mathcal P_i$),  \alg{d-CM}($U_i, Q_i$), and
\alg{Component-Merge}$(M_i, M'_i)$.
The first subroutine constructs a component matching $M_i$ over $\mathcal P_i$,
while the second constructs a $(1,d)$-component matching $M'_i$ over $U_i$ and $Q_i$.
The \alg{Component-Merge}$(M_i, M'_i)$ subroutine performs some low-level communication (described later)
that allows nodes to efficiently learn whether their component merged with other components
by the addition of edges in $M_i$ and $M'_i$ to the forest maintained by the algorithm.
The two matching subroutines run for a fixed round length in $\Theta((D + \sqrt{n})\log{n})$,
while the merge subroutine runs for $\Theta(D+\sqrt{n})$ rounds.
These fixed lengths
allow nodes to remain synchronized during their execution of \alg{MatchingMDST}.

In this section, 
we will analyze \alg{MatchingMDST} under the assumption that these subroutines work correctly.
In particular, we will assume that the matching subroutines always return the correct type of matching,
and with high probability the matching is also {\em maximal}.
In subsequent sections, we will describe and analyze our implementations of these subroutines,
and prove they work correctly with the required probabilities.

\begin{algorithm}[H]
\label{alg:MatchingMDST}
\begin{algorithmic}[1]                    
 \STATE $E_1 := \emptyset$
 \FOR{$i := 1$ \TO $c \log n$} 
  \STATE Let $\mathcal P_i$ be the connected components of $G[E_i]$
  \STATE $M_i := $ \alg{Component-Matching($\mathcal P_i$)}
  \STATE $U_i := \{C \in \mathcal P_i : e \cap C = \emptyset \  \forall e \in M_i\}$ \COMMENT{i.e., components not touched by $M_i$}
  \STATE  $Q_i := \{v \in V : v \not\in \cup_{C \in U_i} C \land \{u,v\} \in E \text{ for some $u \in C$ with $C \in U_i$}\}$ \COMMENT{i.e., vertices not in any $U_i$ component that are adjacent to at least one $U_i$ component}
  \STATE $M'_i := $ \alg{d-CM}($U_i, Q_i$) 
  \STATE $E_{i+1} := E_i \cup M_i \cup M'_i$
  \STATE \alg{Component-Merge}$(M_i, M'_i)$
 \ENDFOR
\RETURN $E' := E_{c \log n}$
\end{algorithmic}
\caption{MatchingMDST}
\end{algorithm}

At a high level, in every iteration of the algorithm we seek to make progress by adding edges which will merge components: if we can reduce the number of components by a constant factor in each iteration, then after $O(\log n)$ iterations we will be left with a spanning tree.  A natural approach is to add matchings, but in order to get running time which is independent of $d$ we need to do slightly more.  We first construct $M_i$, which is intuitively a maximal matching between the components (i.e., a maximal matching in the graph obtained by contracting all of the components of the current subgraph).  This might not include enough components to make significant progress, though, so in the remaining components we try to find a subgraph which has degree at most $d$ and merges a significant number of the remaining components.  This is $M'_i$, which is a maximal $(1,d)$-component matching.  Based on the existence of the optimal (but unknown) spanning tree of degree $d$, we can show that such a maximal subgraph actually touches many of the component nodes, and thus makes progress by merging many of them.

We now analyze this algorithm under the assumption that the subroutines work correctly (as described above).
We begin with a useful property regarding the number of components merged in each phase.

\begin{lemma} \label{lem:large-d-matching}
$|M'_i| \geq  |U_i| / 2$ with high probability.
\end{lemma}
\begin{proof}
With high probability, the matching subroutines return maximal matchings.
Under this assumption,
we first note
that no two components in $U_i$ are adjacent to each other in $G$ (or else $M_i$ would not have been maximal).
Therefore, every component in $U_i$ is adjacent to at least one node in $Q_i$ (or else $G$ would not be connected).  
Let $T$ be an arbitrary spanning tree of $G$ with maximum degree $d$, and let $B \subseteq T$ be the edges of $T$ that have one endpoint in $Q_i$ and the other in a component in $U_i$.  Since $T$ is connected, and has maximum degree $d$, for every component $C \in U_i$ there is at least one edge in $B$ with one endpoint in $C$ and one endpoint in $Q_i$.  
For each $C \in U_i$, select some such edge from $B$ arbitrarily, to create $B' \subseteq B$.  By construction, $B'$ is clearly a $(1,d)$-component matching of size $|U_i|$.  Thus by Lemma~\ref{lem:maximal-approx} we get that $|M'_i| \geq |U_i| / 2$.  
\end{proof}

We now prove that our algorithm efficiently produces a tree with the required degree bound.

\begin{lemma} \label{lem:fast-phases}
With high probability: 
 \alg{MatchingMDST} returns a spanning tree with maximum degree $\hat d = O(d\log{n})$.
\end{lemma}
\begin{proof}
We first prove that \alg{MatchingMDST} always maintains a forest.
In more detail, we prove by induction that $G[E_i]$ is a forest for all $i$.  This is clearly true for $i=1$, since $E_1 = \emptyset$.  Suppose that it is true for some $i$, so we want to show that adding $M_i$ and $M'_i$ to $E_i$ does not result in any cycles.  By definition, $M_i$ is a matching between the connected components of $G[E_i]$, so adding it cannot create any cycles.  When we add $M'_i$, we are adding at most one edge from each component untouched by $M_i$ to a component that was touched by $M_i$, and thus we also do not create any cycles.  It follows that $G[E_{i+1}]$ is a forest.
Our matching routines are always guaranteed to return a matching. The only property that holds probabilistically is their maximality. Therefore, this above observation about maintaining a forest is deterministic.

We now prove that with high probability,   
 $G[E_j]$ has only one component for some $j=O(\log{n})$.  
 Lemma~\ref{lem:large-d-matching} implies that in phase $i$, with high probability at least half of the components in $U_i$ take part in $M'_i$ and thus are joined with at least one other component.
 By definition of $U_i$, any component not in $U_i$ merged during the first matching.  It follows that $|\mathcal P_{i+1}| \leq |M_i| + (|U_i|/2) + \frac12 (|U_i|/2) \leq \frac12 |\mathcal P_i \setminus U_i| + \frac34 |U_i| \leq \frac34 |\mathcal P_i|$ with high probability.  Therefore, after $j = c\log{n}$ phases of  Lemma~\ref{lem:large-d-matching} holding (for appropriate constant $c$),
 we arrive at a single component.
 By a union bound, this lemma holds for the first $j$ phases with high probability.
 \end{proof}

We conclude by noting that under our assumption regarding the correctness  and fixed round complexities of the subroutines,
 Theorem~\ref{thm:bcast-aMDST} follows directly from Lemma~\ref{lem:fast-phases} and the $O(\log n)$ phases of \alg{MatchingMDST}.

\bigskip
{\bf Knowledge of $d$.} \label{sec:knowledge-d} Since the algorithm does not know $d$, it needs to try the values $2, 4, \ldots, n$ in sequence. A value $\hat{d}$ for $d$ succeeds if the maximal $(1, \hat{d})$-component matching actually matched at least half the components in $U_i$, and otherwise it fails.  This can be detected and disseminated using the global BFS tree and aggregation/dissemination strategies discussed in the next section.
If a given estimate $\hat d$ fails, we know that it was too low, so we need not consider it ever again. That is, in the next phase, we continue with the last value of $\hat{d}$ that succeeded.
Therefore, over  $O(\log n)$ phases, we will compute at most $\log n$ total (1, $\hat{d}$)-component matchings that are unsuccessful. This does not impact our asymptotic time complexity.

\subsection{Component Primitives}
\label{sec:logn:3}

Both \alg{MatchingMDST} and the matching subroutines it calls  require the ability to disseminate information within components. 
We implement these abilities with three component primitives: \alg{Component-Broadcast} (which broadcasts a single message throughout a component), 
\alg{Component-Max} (which calculates a max function on values held by nodes in a component), and \alg{Component-Merge} (which updates nodes
within newly merged components, ensuring that at the beginning of each phase, each component has a unique leader, and all nodes in the component
know both this leader and the component size).
The first two primitives are used in both matching subroutines,
while the merge primitive is called at the end of each phase of the \alg{MatchingMDST} algorithm.

In this section we describe the guarantees and implementation details of these  primitives.
All three are deterministic and have a worst case round complexity of at most some $r_{max} = O(D + \sqrt{n})$.

\subsubsection{Preliminaries and Invariants}

Our component primitives maintain the following invariant:
at the beginning of each phase $i$ of \alg{MatchingMDST} (i.e., iteration $i$ of the {\bf for} loop), for each component $C\in \mathcal P_i$:
(a) each $C$ has a unique leader node $ID(C)\in C$; (b) all nodes in $C$ know $ID(C)$; and (c) all nodes in $C$ known $|C|$.
This invariant is trivially satisfied at the beginning of the first phase as all components consist of a single node.
The goal of the \alg{Component-Merge} subroutine called at the end of each phase is to disseminate the appropriate information
to guarantee that the invariant will hold at the beginning of the next phase.

We also assume that at the beginning of the execution nodes construct a BFS $T$ tree over all nodes in the network. Let $u_0$ be the root of this tree.
Using standard synchronous BFS algorithms, this setup requires $O(D)$ rounds. We will use this same tree $T$ throughout the execution.
Without loss of generality, we may assume that each node $u$ knows the height $H(T)$ of the tree as well as level in which $u$ appears.

\textbf{Small and Large Components.}
To ensure efficient round complexities for our broadcast and max primitives, 
we treat {\em small} components (less than $\sqrt{n}$ nodes) differently than {\em large} components (at least $\sqrt{n}$ nodes).  Note that there can be at most $\sqrt{n}$ large components.
Our above invariant ensures that at the beginning of each phase,
each node knows whether it is in a small or large component.

Communication within small components is generally straightforward as we can use a breadth-first tree defined over the component to
efficiently broadcast and convergecast using standard methods.
Large components, by contrast, rely on the global tree $T$.
The key in analyzing the large component primitives will be proving that congestion on $T$ is tractable.

For simplicity, we assume during the execution of these primitives that we run the small component implementations
during even rounds and the large component implementations during odd rounds, preventing interference between the two.

\subsubsection{The \alg{Component-Broadcast} Primitive}
The goal of this primitive is to disseminate a single message through each component: when \alg{Component-Broadcast} is called, we assume at most one node in each component $C\in \mathcal P_i$ has
a message to  disseminate to all nodes in $C$. The primitive disseminates this message to all nodes in $C$.\footnote{Our algorithm
never calls this primitive with more than one node in a component attempting to disseminate a message. For specification completeness,
however, we note that if this primitive is called with {\em multiple} messages within a given component,
our implementation guarantees that each node receives at least one of these messages.}

{\bf Small Components.} This primitive is easy to implement in small components.
Fix some small component $C$. 
Assume some $u\in C$ has a message $m$ to broadcast.
Node $u$ can simply initiate a message flood of $m$ throughout $C$,
where nodes ignore messages broadcast from other components when executing the flood.
This flood requires time $D(C)$, where $D(C)$ is the 
diameter of $C$. Because $C$ is connected and contains less than $\sqrt{n}$ nodes,
we know $D(C) \leq \sqrt{n}$.

{\bf Large Components.}
Large components must share the global tree $T$ to disseminate their messages.
They to so in two steps.
During the first step, nodes route the component messages up $T$ to the root $u_0$.
In each round, each node can send at most one new message to its parent.
A standard pipelining argument, however, establishes that the root will receive all messages within at most 
$H(T) + m$ rounds, where $H(T)$ is the height of $T$ and $m$ is the number of messages.
Because $H(T) \leq D$ and $m \leq \sqrt{n}$ (because there are at most $\sqrt{n}$ large components),
this requires $D+ \sqrt{n}$ total rounds.

At this point, $u_0$ knows all $m$ messages. 
It can disseminate them through $T$ in additional $H(T) + m \leq D+\sqrt{n}$ rounds by pipelining $m$ broadcast waves
down the tree.

\subsubsection{The \alg{Component-Max} Primitive}
This primitive assumes that some subset (perhaps all) of the nodes in each component possess a comparable value of size $O(\log{n})$ bits.
The goal is to compute and disseminate a max function over these values in each component.

{\bf Small Components.}
In each small component $C$,
 the leader $ID(C)$ can execute a standard BFS-based convergecast among nodes in $C$.
 That is, it can initiate a flood that defines a BFS tree in $C$, then the nodes convergecast their values back up to the tree to $ID(C)$.
 This requires $O(D(C))$ rounds, where $D(C) \leq \sqrt{n})$ is the diameter of $C$.

{\bf Large Components.}
Convergecasting is more complicated in large components as potentially multiple such components are using the same global tree $T$ for this purpose,
creating congestion.
\iflong
The first step in our strategy is for each leader of a large component to broadcast its id to all nodes in large components.
We can implement this step in $O(D + \sqrt{n})$ rounds using the \alg{Component-Broadcast} primitive implementation for large components described above. 
This follows because the specific implementation described above goes beyond the specification of the component broadcast
problem to deliver each component's message to {\em all} nodes in the network.

Once all nodes in large components know the complete set of large components, 
the second step is to execute a synchronized convergecast of values from different components over $T$.
This step is easier to describe and analyze if we assume every leaf node in $T$ is at the same depth $H(T)$ (where $H(T)$ is the height/maximum depth of the tree).
If this is not the case, each leaf node $u$ with depth $d(u) < H(T)$ can locally simulate $H(T) - d(u)$ descendants arranged in a line.
Let $T'$ be this resulting tree, made up of real and simulated nodes, that has all leaves at the same depth $H(T)$.

To execute our convergecast, we start every leaf in $T'$ with one token for each of the large components.
Each token is a message that contains the component's leader ID as well as a {\em payload} that holds a value to be convergecast.
The nodes agree on some fixed ordering of these tokens.
They initiate a convergecast up $T'$ for these tokens one by one; i.e., starting the convergecast for the first token in round $1$,
starting the convergecast for the second token in round $2$, and so on. 

For each leaf node $u$ and large component $C$, if $u$ (or the node simulating $u$) is in component $C$ and has a value to convergecast, it puts
its value in the payload for its component $C$ token. Otherwise, it leaves a NIL placeholder in that position.
For each non-leaf node $v$, all tokens for a given large component $C$ will arrive at $v$ during the same round.
Node $v$ calculates the max value among all of these incoming tokens,
as well as its own value (in the case that it is participating in component $C$), and puts this max in the payload of the token for $C$
that it sends to its parent at the start of the next round.

The root $u_0$ of $T$ will receive the convergecast values for all large components after at most $H(T) + n_L$
rounds, where $n_L$ is the number of large components. 
Because $H(T) \leq D$ and $n_L \leq \sqrt{n}$, 
this requires at most $D+\sqrt{n}$ rounds.
At this point, $u_0$ can broadcast all $n_L$ values back down the tree in an additional $D+\sqrt{n}$ rounds
as in the \alg{Component-Broadcast} primitive.
\fi

\subsubsection{The \alg{Component-Merge} Primitive}
This primitive is called at the end of each phase of \alg{MatchingMDST},
after new edges have been selected to be added to the spanning tree.
Each edge connects two previously separate components, requiring them to merge.
The goal of this primitive is to ensure that our component primitive invariants are satisfied after this component merging.
In more detail, for each newly merged component,
we must select a single new leader and ensure all nodes learn this leader and the new component size.

Recall that each phase of \alg{MatchingMDST} executes two matching subroutines.
We handle edges identified by each matching separately. 

\iflong

{\bf Merges from First Matching.} Let $(u,v)$ be an edge added by the first matching. 
This edge requires components $C(u)$ and $C(v)$ to merge.
By the definition of a component matching, these are the only edges adjacent to $C(u)$ or $C(v)$ added by this first matching.
The first step in completing this merge is to select a new leader.
To do so, $u$ can send $v$ the ID of its leader ($ID(C(u))$) and the size of $C(u)$,
and $v$ can send $u$ the ID of its leader $(ID(C(v)))$ and the size of $C(v)$.
Assume that $ID(C(u)) > ID(C(v))$ (the other case is symmetric).
The primitive will elect $ID(C(u))$ to be the leader of the combined component.
Both $u$ and $v$ can send the new leader ID and new component size to all nodes in $C(u)$ and $C(v)$ (respectively), using an instance of \alg{Component-Broadcast}.

{\bf Merges from Second Matching.}
Now consider an edge added by the second matching.
This case is more complicated as the edges included in this matching might
enable many components to merge into a single component.
The details of this second matching, however, provide some useful structure that will aid our merge operations.

In particular, the components participating in this matching are divided into two sets,
which we will call here ${\cal A}$ and ${\cal B}$.
The \alg{d-CM} routine guarantees the follow properties of edges included in the matching it produces:
(1) each edge must have one endpoint in an ${\cal A}$ component and another in a ${\cal B}$ component;
(2) each ${\cal A}$ component contains at most one node that is an endpoint in a matched edge.

Fix some component $C\in {\cal B}$ that must merge with a set ${\cal S}\subseteq {\cal A}$ of components from ${\cal A}$.
Our default rule is that the components in ${\cal S}$ adopt the the leader of component $C$ (i.e., $ID(C)$).
To implement this, we note that for each $C'\in {\cal S}$, there is an edge $(u,v)$ included in the matching
with $u\in C'$ and $V\in C$. Node $v$ knows that its component is in ${\cal B}$,
so it can the ID of its component to $u$, and $u$ can disseminate this through $C'$ using
an instance of \alg{Component-Broadcast}.

At this point, we must also calculate and disseminate the new size of this newly merged component.
To do so, each node in $C$ which is adjacent to at least one other component in $\mathcal S$ in the matching (and at most $d$ such components, since it is a $(1,d)$-matching) can ask its counterparts in ${\cal S}$ for the size of its component.
We can then sum these sizes by running a variation of \alg{Component-Max} in $C$ for these values,
where we replace the max function with the sum operator (the key observation here is that our convergecast strategy
works the same with summing
values as it does for finding the maximum).
It follows that all nodes in $C$ learn the total size the newly merged component (by adding this sum to the size of $C$).
Each endpoint in the matching can pass this information to their counterpart in ${\cal S}$,
which can spread it using another instance of \alg{Component-Broadcast}.

The total cost of executing the merge for edges from both types of matchings
is a constant number of local broadcasts,
and a constant number \alg{Component-Broadcast} and \alg{Component-Max} instances.
The total round complexity is therefore $O(D + \sqrt{n})$.
\fi

\subsection{Matchings Subroutines}
\label{sec:logn:4}

We now describe and analyze the two component matching subroutines called by \alg{MatchingMDST}.
These subroutines also make use of the communication primitives (and the invariants regarding component leaders and sizes maintained by these primitives) discussed in Section~\ref{sec:logn:3}.

\subsubsection{The \alg{Component-Matching} Subroutine}

The \alg{Component-Matching} subroutine
modifies the classical maximal matching algorithm of Israeli and Itai~\cite{II86}.  
As a reminder, 
at a high level, the Israeli and Itai algorithm works as follows  \footnote{We note that although their algorithm is stated for simple graphs, it works equally well in multigraphs. It suffices to revise the definition of 
``good edges'' appropriately.}: 

\begin{itemize}
  \item Stage 1: Each node selects a random incident edge
and proposes it to the other endpoint.
\item Stage 2: Each node that receives a proposal selects a random proposal. 
\item Stage 3: The set of accepted proposals (or chosen edges) induces a graph of degree 2. 
Each node chooses a random incident proposal (either one it proposed, or one it accepted) and tells the other endpoint. If that endpoint also chose that edges, it is included in the matching.
\end{itemize}

We now show how to modify this algorithm to still be efficient when the vertices are actually components, not just nodes.  
Given a collection $\mathcal P = \{C_1, C_2, \dots, C_k\}$ of disjoint sets (components) of vertices, let $C(u)$ denote the cluster containing $u$ for all $u \in \cup_{i=1}^k C_i$.  Consider the following algorithm \alg{Component-Matching}:

\begin{algorithm}[H]
\label{alg:component-matching}
\begin{algorithmic}[1]                    
 \STATE $U := [k]$
 \STATE $M := \emptyset$
 \FOR{$i := 1$ \TO $c \log n$}
 \STATE \texttt{//Stage 1}
 \STATE Every $u$ in each component $C_i$ with $i \in U$ assigns a random priority value in $[n^3]$ to each edge from $u$ to a different component.  Let $r_u$ be the maximum of these priority values, corresponding to edge $e_u$.
  \STATE Run \alg{Component-Max($\{C_i : i \in U\}$)} with values $r_u$ to find the highest priority edge leaving each remaining component.  For component $C_i$, let this edge be $e_u$ where $u \in C_i$.
  \STATE $u$ sends a proposal to the other endpoint of $e_u$. 
  \STATE \texttt{//Stage 2}
  \STATE Every node receiving a proposal assigns each received proposal a random priority in $[n^3]$.  If $u$ is such a node, let $p_u$ denote the largest of these priority values.
  \STATE Run \alg{Component-Max($\{C_i : i \in U\}$)} with values $p_u$ (if $p_u$ not defined, set it to $-\infty$ first) to find the highest priority incoming proposal in each remaining component. For component $c_i$, let the edge corresponding to this proposal be between $v_i \in C_i$ and $u_i \not\in C_i$
  \STATE $v_i$ sends an ``accept" message to $u_i$
  \STATE \texttt{//Stage 3}
  \STATE If $u$ receives an ``accept" message from $e_u$, use \alg{Component-Broadcast} to send this to the leader of $C(u)$
  \STATE The leader of each cluster $C_i$ now knows whether $C_i$ sent a proposal which was accepted and whether $C_i$ accepted a proposal from another cluster.  If only one of the two, let $e_i$ be this edge.  If both, the leader chooses one of the two edges at random to be $e_i$.  The leader broadcasts the identity of this edge to all of $C_i$ using \alg{Component-Broadcast}.
  \STATE The endpoint $u_i$ of $e_i$ that is inside $C_i$ sends a commit message to the other endpoint of $e_i$.  If the other endpoint also sends a commit message to $u_i$, then we add $e$ to $M$, send this message to all of $C_i$ using \alg{Component-Broadcast}, and remove $i$ from $U$.  
 \ENDFOR
 \RETURN $M$
\end{algorithmic}
\caption{Component-Matching($\mathcal P = \{C_1, \dots, C_k\}$)}
\end{algorithm}

We now analyze this subroutine:

\begin{lemma} \label{lem:maximal-component-matching}
Let $G = (V, E)$, and let $\mathcal P = \{C_1, C_2, \dots, C_k\}$ be a collection of disjoint sets of vertices (components) such that $G[C_i]$ is connected for all $i \in [k]$.  The \alg{Component-Matching($\mathcal P$)} subroutine terminates in $O((D+\sqrt{n}) \log n)$ rounds and returns a component matching of $\mathcal P$.
With high probability, this matching is maximal.
\end{lemma}
\iflong
\begin{proof} 
By construction there are $O(\log n)$ iterations in \alg{Component-Matching}, so to bound the running time we just need to argue that each iteration takes at most $O(D + \sqrt{N})$ time.  This follows directly from the analysis of the \alg{Component-Broadcast} and \alg{Component-Max} component primitives in Section~\ref{sec:logn:3}.

It follows from the definition of this algorithm and the correctness of the component primitives, that it always returns a matching.
We are left therefore to
prove with high probability that this matching is {\em maximal}.
To do so, we can defer to the analysis of~\cite{II86}.
In particular,
notice that \alg{Component-Matching} exactly mimicks the II algorithm in the graph obtained by contracting every component to a single node.  So since after $O(\log n)$ iterations the II algorithm has returned a maximal matching with high probability~\cite{II86}, \alg{Component-Matching} returns a maximal component matching with high probability.
\end{proof}
\fi

\subsubsection{The \alg{d-CM} Subroutine}

We now analyze the \alg{d-CM} subroutine, which computes $(1,d)$-component matchings.
As in the case of standard component matchings, 
we design our algorithm for $(1,d)$-component matchings by generalizing a classical algorithm to also work for components.  
In this case, we modify a maximal matching algorithm of Luby \cite{Luby86} (which more generally produces maximal independent sets) 
for the bipartite graph setting in which we will compute our $(1,d)$-component matchings.

We first describe this classical algorithm, before giving our generalization.
Luby's algorithm runs in phases, each of which runs on the subgraph containing the nodes that are not yet matched and the edges connecting unmatched nodes. Each phase proceeds in two stages on a bipartite graph with parts $U$ and $Q$.
\begin{itemize}
  \item Stage 1: Each node $u$ in $U$ assigns each incident edge $e$ a random priority value $r_e$ chosen from $[1,n^3]$.
It determines the incident edge $e=(u,w)$ with highest priority and proposes it by broadcasting its label and priority.
\item Stage 2: Each node $w$ in $Q$ that receives a proposal chooses the one with the highest priority and adds to the matching.
\end{itemize}
The effect is that an edge is chosen if its random value is locally maximum, i.e., exceeding that of all its neighbors. That is how Luby's algorithm is normally described \cite{Luby86}, and it is known that the algorithm runs in $O(\log n)$ rounds.

We now describe our \alg{d-CM} subroutine which generalizes the above strategy to our setting, where $U$ is a set of components and we are trying to compute a $(1,d)$-component matching.  Intuitively, we just use our communication primitives to allow components in $u$ to act as if they were nodes (at a time complexity cost of $O(D+ \sqrt{n})$), and we allow nodes in $Q$ to accept up to $d$ proposals rather than $1$. 

\begin{algorithm}[H]
\label{alg:d-CM}
\begin{algorithmic}[1]                    
 \STATE $E_1 := \emptyset, A := U$, and $B := Q$.
 \FOR{$i := 1$ \TO $c \log n$}
  \STATE \texttt{//Stage 1}
  \STATE Every node in $B$ sends a message to its neighbors announcing that it is in $B$.
  \STATE Every node $u$ which is in some component in $A$ receives these messages and so learns of its neighbors in $B$.  $u$ then assigns a value $r(e) \in [n^3]$ to each edge $e = \{u,v\}$ with $v \in B$ chosen uniformly at random from $[n^3]$.
  \STATE Run \alg{Component-Max} in every component $C_i \in A$ to select the edge $e_i$ from $C_i$ to $B$ with maximum assigned value.  
  Let $u$ be the endpoint of $e_i$ in $C_i$, and let $v$ be the endpoint of $e$ in $B$.  Then $u$ sends a ``proposal" along $e_i$ to $v$ which contains the value $r(e_i)$.
  \STATE \texttt{//Stage 2}
  \STATE For every $v \in B$, let $p(v)$ denote the number of proposals that it just heard and let $m(v)$ denote the number of edges in $E_1 \cup E_2 \cup \dots \cup E_i$ incident on $v$.  Then $v$ ``accepts" the $f(v) = \min(p(v), d - m(v))$ proposals by sending the value of $r(e')$ to its neighbors, where $r(e')$ is the $f(v)$'th largest proposal that $v$ just heard.  Let $E_{i+1}$ be the set of edges that were just accepted by a node in $B$.  If $f(v) = d- m(v)$ then $v$ removes itself from $B$.
  \STATE Every vertex $u \in C_i \in A$ that sent a proposal now knows if its proposal was accepted, by checking whether the value of the edge it proposed is at least the value returned by the endpoint in $B$.  Run \alg{Component-Broadcast} to disseminate this information in each $C_i \in A$.  Any $C_i \in A$ who had a proposal accepted now removes itself from $A$ (all of the vertices in $C_i$ know that it had a proposal accepted and so they do not participate in future rounds).
 \ENDFOR
 \RETURN $\cup_{i=1}^{c\log n} E_i$
\end{algorithmic}
\caption{d-CM($U, Q$)}
\end{algorithm}

\begin{lemma} \label{lem:d-component-matching-algorithm}
Let $U = \{C_1, C_2, \dots, C_k\}$ be a collection of connected components.  Let $Q \subseteq V$ be a set of vertices.  Subroutine \alg{d-CM($U, Q$)} computes a
$(1,d)$-component matching of $(U,Q)$ in $O((D + \sqrt{n}) \log n)$ rounds.
With high probability, the matching is maximal.
\end{lemma}
\begin{proof}
It is easy to see by induction that \alg{d-CM} always maintains a $(1,d)$-matching.  So we just need to prove that it is maximal after $O(\log n)$ rounds, with high probability.  
To see this, note that a $(1,d)$-component matching is equivalent to an ordinary matching in a \emph{replicated graph} $H'$ which contains $d$ copies of each node in $Q$ with each copy retaining all the incident edges of the original. We argue that the solution found in each phase of our algorithm dominates the solution found by Luby's algorithm on the replicated graph, where each component in $U$ sends separate proposals to each of the $d$ copies.

First, observe that a node $v \in Q$ accepts as least as many proposals
in a phase of our algorithm as the $d$ copies do in a phase of Luby.
Second, each proposal of a neighbor of $v$ is equally likely to be accepted.
Thus, the solution found by our algorithm stochastically dominates the one by Luby on the replicated graph.  Since Luby's algorithm in the replicated graph terminates in at most $O(\log(nd)) = O(\log n)$ rounds with high probability, after $O(\log n)$ rounds our algorithm will have found a maximal $(1,d)$-component matching with high probability.   

To achieve the final time complexity, 
we note that each iteration of the main loop in our algorithm makes a constant number
of calls to the component communication primitives. 
As established in Section~\ref{sec:logn:3},
each such call requires $O(D + \sqrt{n})$ rounds.
\end{proof}


\section{Improved Approximation}
\label{sec:slow}

We give a local-improvement algorithm in broadcast-CONGEST that produces a spanning tree of degree $O(d+\log n)$. 
The algorithm can be used as a post-processing phase, and can also be viewed as an \emph{anytime} algorithm: the execution can be stopped after any phase with a valid and improved solution, if needed.
The running time depends on the initial tree that is fed into the improvement algorithm, and if we first run \alg{MatchingMDST} and use the output as the starting tree to this algorithm, then the total time complexity of the algorithm is $O((D+\sqrt{n}) \log^4 n)$.

\textbf{Overview:} The algorithm borrows the improvement idea from F\"urer and Raghavachari's~\cite{fr94} sequential algorithm. Their algorithm, however, tries to completely eliminate all maximum degree vertices, which can only be achieved by a recursive process that is difficult or impossible to parallelize. We instead aim to find only the ``nice'' improvements that can be easily processed, and as a result, can be performed in parallel. This results in gradual decrease of high degree vertices, until a few types of degrees remain.
To speed up the convergence of the process, the algorithm also tries to substitute only edges whose endpoints have very low degree.

\subsection{Parallel Improvements}

We argue in this subsection that many improvements can be made in parallel, 
under the right conditions.
Let $T$ be the input spanning tree and let $d_T(v)$ denote the degree of node $v$ in $T$.
Let $h > 2d$ be a number to be determined.
Let $\gamma,\gamma_0$ be numbers such that $\gamma > \gamma_0 \ge h$.
We aim to reduce the number of vertices of degree $\gamma$ or more, but only by increasing the degrees of nodes of degree less than $\gamma_0$.
Let $X_q$ be the set of nodes of degree at least $q$, for integer $q$.

We root $T$ from an arbitrary node in $X_\gamma$.
Removing the nodes in $X_\gamma$ from $T$ results in a collection of rooted trees which we shall call \emph{branches}. 
A branch is a \emph{leaf branch} if no other branches are contained in its subtree, and otherwise is an \emph{internal branch}.
The root of a branch is the root of corresponding subtree in $T$.
Branches with the same parent are collectively called a \emph{bundle}. 
A \emph{leaf bundle} is a bundle that contains at least one leaf branch.
The \emph{parent} of a branch is the parent of the branch root.
For a branch $B$, denote the edge from its root to its parent as $e(B)$.
For a node $u$, let $B_u$ denote the branch containing $u$.
We shall overload set names to also refer to the sizes of those sets.

For a directed or \emph{oriented} edge $(u,v)$, we refer to $u$ ($v$) as its \emph{source} (\emph{destination}), respectively.
Let $h(e)$ be the source of an oriented edge $e$.
Orientations are considered here only to clarify how improvements are applied.

\begin{definition}
An oriented subgraph $M$ of $G$ is \emph{valid} if $T_M = (T \setminus M')\cup M$ is an (undirected) tree, where $M' = \{ e(B_{h(e)}) : e \in M\}$.
\end{definition}

The idea is to replace the parent edges of some leaf branches with edges in $M$ so as to reduce the degrees of these parents.

We also want the resulting degrees in $T_M$ to be ``better'' than before.
We say that a parent of a leaf branch $B$ is \emph{improved} if $B=B_{h(e)}$ for some $e \in M$. Namely, if its edge to the branch will be removed as part of the improvement, and its degree therefore reduced.

\begin{definition}
A valid subgraph $M$ is an \emph{$(x,y)$-improvement} if:
a) each improved parent $v$ has $d_{T_M}(v) \ge x$ (i.e., $v$ is not improved too much), 
b) each node $v$ with $d_{T_M}(v) > d_T(v)$ has $d_{T_M} \le y$ (i.e., low degree nodes cannot get too much worse).
\end{definition}

An oriented edge is \emph{good} if its source is in a leaf branch and its destination in a different branch (not necessarily a leaf branch), and both endpoints have degree less than $\gamma_0$ in $T$.
Our parallel improvement strategy is built on the following observation.

\begin{observation}
Let $M$ be a subgraph of good oriented edges such that
each branch has at most one outgoing edge, and
no branch is both the source and destination of edges in $M$.
Then $M$ is valid.
If, additionally, each node $v$ incident on an edge in $M$
satisfies $d_M(v) \le q$ and each bundle has at most $q$ outgoing incident edges of $M$, then $M$ is a $(\gamma-q, \gamma_0+q)$-improvement. 

\label{obs:parimp}
\end{observation}

\begin{proof}
Recall $M' = \{ e(B_{h(e)}) : e \in M\}$ and consider the edges to be oriented from branch roots to their parents.
Removing $M'$ breaks $T$ into $M'+1$ components: a leaf branch for each source of an edge in $M$, and $\hat{T}$ (the rest). 
Observe that 
the sources of edges in $M$ are in the same branches as the sources in $M'$.  
Since no branch is both the source and destination of edges in $M$,
the destinations of all edges in $M$ are in $\hat{T}$. 
Hence, adding $M$ back in reconnects the tree.

Since each bundle has at most $q$ outgoing edges in $M$, 
nodes in $X_\gamma$ have their degree decreased by at most $q$. Also, since $M$ has $d_M(v) \le q$, its endpoints increase their degree by at most $q$. Since they were all of degree less than $\gamma_0$, no vertex of degree $\gamma_0$ or more in $T$ is of higher degree in $T'$. 
\end{proof}

\mypar{Distributed Improvement Algorithm}
We encode this observation in an algorithm \alg{Improve} with parameters $\gamma,\gamma_0,q$,
which takes the tree $T$, finds a $(\gamma-q,\gamma_0+q)$-improvement, and produces a modified tree $T_M$.
The algorithm proceeds as follows.

Form the bipartite graph $H = (U,Q,E')$, where $U$ is the set of leaf branches, and $Q$ is the set of nodes in $V \setminus X_{\gamma_0} = \{v \in V : d_T(v) < \gamma_0 \}$ with an incident edge to a leaf branch.
For every edge in $G$ between endpoints of degree less than $\gamma_0$, at least one of which is in a leaf branch, 
there is an edge in $H$, which we view as being oriented from $U$ to $Q$.
If both endpoints are in leaf branches and have degree less than $\gamma_0$, then the edge appears twice, 
once in each direction.

We now find a near-maximum \emph{constrained $(1,q)$-matching} $\hat{M}$ in $H$,
which is a $(1,q)$-matching with the additional constraint that at most $q$ edges are 
outgoing from any leaf bundle in $U$. \iflong We do this with a procedure \alg{Constrained-Matching} which we discuss in more detail in Appendix~\ref{app:cm} and in the next section on implementation.
\fi

Each leaf branch $B$ has at most one outgoing edge in $\hat{M}$.
For each leaf branch $B$ with at least two incoming edges in $\hat{M}$, we remove the outgoing edge from $B$ in $\hat{M}$ (if it exists).
If a leaf branch has exactly one incoming and one outgoing edge, then it removes one of them at random.
Let $\bar{M}$ denote the resulting subgraph and observe that it satisfies the prerequisites for Observation~\ref{obs:parimp}, and is therefore a valid $(\gamma-q, \gamma_0+q)$-improvement.

\mypar{Implementation of \alg{Improve}}
The nodes first use intra-component communication (\alg{Component-Max} and \alg{Component-Broadcast}) to compute several properties:
a) Determine their branch id, which is the node of the highest id in that branch; b) Determine if a branch is a leaf branch, equivalently if only one tree edge exits the branch; c) Learn the id of the branch root, and its parent, the root of the bundle.

\iflong
In order to find a near-maximum constrained $(1,q)$-matching, we design an algorithm \alg{Constrained-Matching} which we describe in detail in Appendix~\ref{app:cm}.  Note that the difference between a constrained $(1,q)$-matching and a $(1,d)$-matching (as discussed in Section~\ref{sec:logn:4}) is the extra constraint that each bundle can only have $q$ incident edges on its leaf branches.  To overcome this extra difficulty, we design a very different algorithm based on finding maximal flows in an auxiliary graph related to $H$.  As with \alg{d-CM}, one set of nodes in this auxiliary graph corresponds to components, but by using the communication primitives from Section~\ref{sec:logn:3} we can treat these components simply as vertices by spending $O(D+\sqrt{n})$ time.  We prove in Appendix~\ref{app:cm} that this algorithm takes $O(\log n)$ time (so $O((D+\sqrt{n}) \log n)$ time when using the communication primitives) and computes a $128$-approximation to the maximum constrained $(1,q)$-matching (i.e., it constructs a constrained $(1,q)$-matching with at least $1/128$ as many edges as the maximum constrained $(1,q)$-matching).  Note that unlike our previous matching algorithms, \alg{Constrained-Matching} does not compute a maximal solution; it instead computes a maximal \emph{fractional} solution and then rounds this fractional solution (all in a distributed fashion).
\fi

\mypar{Analysis}
We first argue that every maximal constrained $(1,q)$-matching must have many edges.
We first need an accounting of the adjacencies of nodes in $X_\gamma$ that do not contribute to that count.

\begin{lemma}
  At most $2(X_\gamma-1)$ adjacencies of nodes in $X_\gamma$ are not to leaf branches.
\label{lem:nrleafs}
\end{lemma}

\begin{proof}
Adjacencies of a node in $X_\gamma$ are either to a leaf branch, an internal branch, or to another node in $X_\gamma$. We bound the latter two.

Let $s$ be the number of nodes in $X_\gamma$ that have another node in $X_\gamma$ as parent, and $r$ be the number that have an internal branch as parent.
Then, $s+r = X_\gamma-1$, as only the root satisfies neither.
Equally many adjacencies of nodes in $X_\gamma$ will be to a child that is an internal branch or another nodes in $X_\gamma$.
\end{proof}

\begin{lemma}
$\hat M \geq \frac{q}{128\gamma} ((\gamma -2)X_\gamma - dX_{\gamma_0})$.
\label{lem:genimps}
\end{lemma}

\begin{proof}
We first show that there exists a large constrained $(1,q)$-matching in $H$, and then use that fact that $\hat M$ is a 64-approximation. 

We restrict our attention to a smaller subgraph.
From each leaf bundle with $s$ leaf branches, 
retain an arbitrary set of $\min(s, \gamma)$ leaf branches, and
let $L$ denote the resulting set of leaf branches.
By Lemma \ref{lem:nrleafs}, 
  $L \ge \sum_{v\in X_\gamma} \min(d_T(v),\gamma) - 2(X_\gamma-1)
   = (\gamma-2)X_\gamma+2$.

Let $OPT$ be a spanning tree of maximum degree $d$, rooted at an arbitrary node in $X_\gamma$.
For each leaf branch in $B \in L$, let $v_B$ be a node in $B$ of maximal height in OPT, and let $e(v_B)$ be the edge to its parent in $OPT$.
Let $R = \{e(v_B) : B \in L\}$.
Then, $R = L \ge (\gamma-2)X_\gamma + 2$.
Since $OPT$ has maximum degree $d$, at most $dX_{\gamma_0}$ edges in $R$ have at least one endpoint in $X_{\gamma_0}$.
Let $R' \subseteq R$ be the set of edges with both endpoints of degree less than $\gamma_0$ in $T$. Then, 
$R' \ge R-dX_{\gamma_0} \ge (\gamma -2)X_\gamma - dX_{\gamma_0}$. 

The resulting subgraph of $H$ is a $(1,d)$-matching (since OPT has maximum degree $d$ and we chose at most one edge out of each leaf branch).
If we contract all the leaf branches in a bundle into a single node,
we obtain a bipartite subgraph of maximum degree at most $\gamma$ 
(since at most $\gamma$ leaf branches were retained from each bundle and $\gamma > d$). 
This can be $\gamma$-edge colored, and hence it contains a $q$-matching
of size at least $\frac{q}{\gamma} R'$, corresponding to a constrained $(1,q)$-matching.

The fact that $\hat M$ is a 128-approximation to the maximum constrained $(1,q)$-matching now implies the lemma.
\end{proof}

To turn the matching $\hat{M}$ into a valid subgraph means shedding some edges to get $\bar M$, but a constant fraction must remain. 

\begin{lemma}
  $\E[\bar{M}] \ge \hat{M}/8$.
\label{lem:barm}
\end{lemma}
\iflong
\begin{proof}
Recall that each leaf branch has at most 1 outgoing edge in $\hat{M}$.
At most $\hat{M}/2$ leaf branches have two or more incoming edges in $\hat{M}$ and thus at most $\hat M / 2$ branches remove their outgoing edge due to multiple incoming edges.   
For the remaining edges, they have probability at least half of not being removed by its source (destination), respectively, so survive that selection with probability at least $1/4$. 
\end{proof}
\fi

The following is the key condition for finding large parallel improvements.

\begin{theorem}
Let $c$ be a constant and $q$ be a parameter.
If $X_{\gamma_0} \le c \cdot X_\gamma$,
then \alg{Improve}$(\gamma,\gamma_0, q$) yields a valid $(\gamma-q,\gamma_0+q)$-improvement containing $\Omega(q \cdot X_\gamma)$ edges in expectation, for appropriately chosen $h = h_c = \Theta(d)$. 
\label{thm:sizem}
\end{theorem}

\begin{proof}
Recall that $\bar{M}$ is a valid improvement $(\gamma-q, \gamma_0 + q)$-improvement by Observation~\ref{obs:parimp}, and by Lemma \ref{lem:barm} is of expected size at least $\hat{M}/8$.
By Lemma \ref{lem:genimps}, the hypothesis, and the fact that $\gamma \ge h$, we get that
$128\hat{M}/q 
  \ge (1-2/h)X_\gamma - dX_{\gamma_0}/h 
   \ge X_\gamma\left(1 - \frac{2+dc}{h}\right)$. 
Now we choose $h = h_c = (dc+2)/(1-\delta)$, for any $\delta > 0$.
Then $128\hat{M} \ge \delta q X_\gamma$, and thus by Lemma~\ref{lem:barm} we get that $\E[\bar{M}] \geq \Omega(q \cdot X_{\gamma})$ . 
\end{proof}

\subsection{Repeated Improvements}

Theorem~\ref{thm:sizem} allows us to find large improvements under certain assumptions ($X_{\gamma_0} \le c \cdot X_\gamma$).  But now we need to show how to repeatedly find improvements in a smart way, so we make significant progress on decreasing the degrees in the tree.  Our algorithm \textsc{Rehab} takes a parameter $z$ and works as follows.  

Let $k$ be the (current) maximum degree of the tree that we are working on.
Let $b_j =  h + j\cdot z$, for $j \ge 0$.
Let $C_j = X_{b_j}$ denote the \emph{blocks}, which are sets of nodes of degree at least $b_j$, for $j\ge 0$.
Let $\tau = 2/(1-\delta)$, for some fixed $\delta$.
Define $\sigma_j = \tau^j$, for $j \ge 0$.

 \begin{algorithm}[H]
\label{alg:Rehab}
\begin{algorithmic}[1]                    
 \STATE Let $b_j =  h + j \cdot z$, for $j \ge 0$
 \STATE $j:=2$
  \REPEAT
    \STATE \alg{Improve}$(b_j, b_{j-2},z)$.
    \STATE $j := \argmax_s C_s \sigma_s$
  \UNTIL{$j\le 1$}
\end{algorithmic}
\caption{Rehab($z$)}
\end{algorithm}

To implement this algorithm, we need to compute the sizes of the blocks $C_s$
and disseminate, from which the next value of $j$ can be determined by each node. This can be done by a count-aggregation on a global BFS tree. We show later that there are always only $O(\log n)$ non-empty blocks, which allows to compute this in time $O(D+\log n)$.

The convergence or termination of the algorithm is not obvious, but will be derived shortly.
The key property of the algorithm is that when it terminates, the blocks $C_j$ must have geometrically decreasing cardinalities.

\begin{observation}
When \alg{Rehab} terminates, 
$C_j \sigma_j \le \max(C_0\sigma_0, C_1\sigma_1) \le \tau n$, for all $j\ge 1$.
Thus, each $C_j$ contains at most $n/\sigma_{j-1} = n/\tau^{j-1}$ nodes,
and each $C_j$ with $j \geq \log_{\tau} n + 1$ contains no vertices.
Hence, the maximum degree of the resulting tree is bounded by $h + z \log_\tau n$. 
\label{obs:xxx}
\end{observation}

We proceed in a series of \emph{epochs}, where in each we run the \alg{Rehab} algorithm with progressively finer block-sizes. 

 \begin{algorithm}[H]
\label{alg:Epochs}
\begin{algorithmic}[1]                    
  \STATE $i := 2$
  \REPEAT
    \STATE $z_i := \lceil (k-h)/2^i \rceil$
    \STATE \alg{Rehab}$(z_i)$
    \STATE $i := i+1$
  \UNTIL $z_i=1$
\end{algorithmic}
\caption{Epochs}
\end{algorithm}

We make progress arguments in terms of a potential function $w$.
We define the weight $w(v)$ of a node $v \in C_j\setminus C_{j+1}$ in epoch $i$ to be
\[ w(v) = 1 + (d_T(v)-b_j)\sigma_j + \sum_{s=0}^{j-1} (b_{s+1}-b_s)\sigma_s
      = 1 + (d_T(v)-b_j)\sigma_j + z_i \sum_{s=0}^{j-1} \sigma_s. \]
Namely, each adjacency contributes a $\sigma$-term to the weight, with the terms increasing by a factor of $\tau$ as we move past each threshold $b_j$.
      
Observe that if nodes $v$ and $v'$ are in $C_j \setminus C_{j+1}$, then
$w(v) = \Theta(w(v')) = \Theta(z \sigma_j) = \Theta(z\sum_{s=0}^j \sigma_s) $.
The weight of the whole instance is $w = \sum_{v \in V} w(j)$.

\begin{lemma}
Let $j$ be the index that maximizes $C_j \sigma_j$.
If $j > 1$, then the call to \alg{Improve}$(b_j, b_{j-2},q)$ yields weight decrease $\Omega(w/t)$, where $t$ is the number of non-empty blocks.
\label{lem:bestimp}
\end{lemma}
\begin{proof}
We first claim that each edge $e=(u,u')$ of the subgraph $\bar{M}$ contributes a 
drop of $\Omega(\sigma_j)$ in the total weight.
Namely, it was used to decrease the degree of a node in $C_j$, 
for a weight decrease at least $\sigma_{j-1}$, while the increase in 
the weights of $u$ and $u'$ is at most $2\sigma_{j-2}$.
The net decrease is then $\sigma_{j-1} - 2\sigma_{j-2} \le \sigma_{j-1}(1-2/\tau)
= \delta \sigma_{j-1} = \Omega(\sigma_j)$. 

By assumption, $C_{j-2} \le C_j \sigma_j/\sigma_{j-2} = C_j \tau^2$.
So the hypothesis of Theorem~\ref{thm:sizem} holds for $\gamma=b_j$ and $\gamma_0 = b_{j-2}$, where $c=\tau^2$.
Observe that $h=h_c = 4d (1+O(\delta))$.
By Observation \ref{obs:parimp} and Theorem~\ref{thm:sizem},
the expected number of improvements is $\E[\bar{M}] = \Omega(q C_j)$.
Hence, using the above claim on the impact of a single improvement,
the total weight decrease is $\Omega(q C_j \sigma_j)$.

Observe that $w=\sum_v w(v) = \Theta(\sum_j (C_j - C_{j+1}) \cdot z \sum_{s=1}^j \sigma_s) = \Theta(\sum_j C_j \cdot z \sigma_j)$.
Since $j$ maximized $C_j\sigma_j$, the call to \alg{Improve} yields an improvement of $\Omega(q/z \cdot w/t) = \Omega(w/t)$.
\end{proof}

This now lets us bound the total time complexity.

\begin{lemma} \label{lem:epochs-time}
The time complexity of \alg{Epochs} is $O((D+\sqrt{n})\log^4 n)$.
\end{lemma}

\begin{proof}
Refer to each iteration of \alg{Rehab} as a \emph{phase}. 
Each phase takes $O((D+\sqrt{n})\log n)$ steps: a call to \alg{Improve} (which we argued takes at most $O((D + \sqrt{n})\log n)$ rounds,
and $O(D+\log n)$ steps to determine the next $j$.

By Observation~\ref{obs:xxx} and the fact that $w(v) = \Theta(\sigma_j z)$, for a node $v \in C_j \setminus C_{j+1}$, the weight of each block $C_j$ is $O(zn)$ at the end of an epoch.
Thus, the total weight at the end of each epoch is $O(tzn)$, which we can crudely bound by $O(n^3)$.
By halving the value of $z$, the weight of each node is at most squared.
Thus, the total weight at the start of an epoch is also at most squared or $O(n^6)$. 

Each phase reduces the weight by a fraction $\Omega(1/t)$,
where $t$ is the number of non-empty blocks $C_j$).
An epoch starts with total weight $O(n^6)$ and ends with weight at least $n$ (since the minimum weight of a node is 1).
Thus, the number of phases in an epoch is $O(t \log(n^6))$.
By Obs.~\ref{obs:xxx}, $t = \min(2^i, \log n)$ in epoch $i$.
Hence, the total number of phases is on the order of
\begin{align*}
 \sum_{i=1}^{\log(k-h)} &\min(2^i, \log n) \log n = \sum_{i=1}^{\log\log n} 2^i \log n + \sum_{i=\log\log n}^{\log k} \log^2 n \\
 & = \log^2 n + (\log k - \log\log n)\log^2 n = O(\log(k/\log n) \log^2 n)\ .
\end{align*}
Note that maximum degree of $T$ can go down as the algorithm progresses but it never increases, thus we can conservatively work with the original maximum degree $k$.
When $k=O(d\log n)$, this results in  $\log d \log^2 n = O(\log^3 n)$ phases.
Hence, the total time complexity is $O((D+\sqrt{n})\log^4 n)$.
If we started with an arbitrary spanning tree rather than the output of \alg{MatchingMDST}, the time complexity would be one more logarithmic factor.
\end{proof}

\begin{theorem}
\alg{Epochs} returns a spanning tree with maximum degree at most $O(d + \log n)$ in at most $O((D + \sqrt{n})\log^4 n)$ rounds in the broadcast-CONGEST model.
\end{theorem}
\begin{proof}
The running time is direct from Lemma~\ref{lem:epochs-time}, and the degree bound is implied by Observation~\ref{obs:xxx} and the fact that \alg{Epochs} eventually calls \alg{Rehab} with a constant parameter. 
\end{proof}




\bibliographystyle{plainurl}

\bibliography{refs}

\iflong
\appendix

\section{Algorithm for Constrained Matchings}
\label{app:cm}

We give here a randomized distributed algorithm in the broadcast-CONGEST model for finding near-maximum $(1,q)$-constrained matchings, running in $O(\log n)$ time.
The algorithm is based on finding approximate fractional matchings, viewing it as a flow in a shallow network.

\subsection{Algorithm}
Constrained matchings correspond to flows in a related flow graph $F$.
The vertices of the flow graph $F$ consist of
the set $B$ of leaf bundles, the extreme nodes $s$ and $t$,
as well as the sets $U$ and $Q$ from the graph $H$.
There is a directed edge from $s$ to each leaf bundle node, from each leaf bundle to its
constituent leaf branches, from leaf branches to the nodes in $Q$ they
are adjacent to in $H$, and finally from each node in $Q$ to $t$. Edges from
$s$ to leaf bundles, and those from $Q$ to $t$, are of capacity $q$,
while the rest are of unit capacity.

Observe that there is a one-one correspondence between (fractional)
$q$-constrained matchings in $H$ and flows in $F$.
Each edge in $H$ has a unique flow path in $F$ and vice versa.
Thus we may specify a flow in $F$ by giving only the flow on edges in $H$.
A maximal fractional matching corresponds to \emph{maximal flow},
where every $s-t$ path has some node that is \emph{saturated}, i.e., whose flow is at full capacity.

For a flow $f$, let $f(e)$ denote the flow through edge $e$, $f(v)$ denote the flow going out of $v$, and $v(f) = f(s)$ be the \emph{value} of the flow.
The \emph{size} of a fractional matching equals the value of the corresponding flow.
We say that a node is \emph{full} if it has incoming flow at least 1/8-th of its capacity.

The algorithm initially assigns each flow path a flow of $1/m$, where $m$ is the number of edges. In each round, every non-full node in $Q$ doubles the flow on its incident paths from non-full nodes in $U$ (through its parents in $B$) by sending a ``double the flow" message to its neighbors (full neighbors will ignore this message). By using \alg{Component-Broadcast} appropriately, it is easy to see that each leaf branch in $U$ can compute the total incident LP value, and so can every leaf bundle, so every leaf branch in $U$ knows whether it is full.  It sends this information to its neighbors in $Q$, which will then know which of its incident edges did actually double the flow.  Then $Q$ can begin the the next round of the algorithm.  After $O(\log n)$ rounds there will be no way of sending more flow using only non-full nodes (as we show in the next subsection), so after $O(\log n)$ rounds we more to the next part of the algorithm, where we use randomized rounding to find a constrained $(1,q)$-matching.  

In particular, we would like to do the following (from a centralized perspective).  We would first add every edge $e$ from $U$ to $Q$ to a set $S$ independently with probability $f(e)$.  Then any leaf branch in $U$ with more than one incident edge in $S$ removes all such edges from $S$, any leaf bundle with more than $q$ incident edges in $S$ removes all such edges from $S$, and any node in $Q$ with more than $q$ incident edges in $S$ removes all such edges from $S$.  This would by construction result in a constrained $(1,q)$-matching, which we call $S'$.  

In order to implement this in broadcast-CONGEST, we first have every vertex $v$ in each leaf branch in $U$ make the appropriate randomized decisions for the edges from $U$ to $Q$ that are incident on $v$, so every vertex $v$ in each leaf branch known which edges of $S$ are incident on it.  Note that this results in the same $S$ as in the centralized algorithm.  Now if $v$ added more than one incident edge to $S$, then it removes all of these edges from $S$.  Otherwise, if $v$ added exactly one edge to $S$, it broadcasts the identity of this edge to all of its neighbors (and in particular the other endpoint of the added edge) as well as using \alg{Component-Broadcast} to send the identity of the edge to the rest of the leaf branch containing it.  If in some leaf branch in $U$ there are multiple \alg{Component-Broadcast} instances occurring, then any node which detects this sends a ``failure" message through the branch (using another \alg{Component-Broadcast}), and all nodes in the branch remove all of their incident nodes which were in $S$ from $S$.  

Now each leaf branch has either $0$ or $1$ incident edge in $S$, and all nodes in the branch know the identity of this edge (if it exists).  The root of each leaf branch sends to its parent the identity of this edge.  So now the root of each leaf \emph{bundle} knows the edges incident on the bundle which are in $S$.  If there are more than $q$ such edges, then this bundle root removes them all by using \alg{Component-Broadcast} to send a message to all of the leaf branches in the bundle.  Similarly, each node in $Q$ knows all of its incident edges that are in $S$, and if there are more than $q$ of them then it removes all of them from $S$ by broadcasting a failure message to its neighbors.  The edges which survive this process are $S'$, and it is easy to see that it is precisely the same set as would have been computed in the centralized version.

This completes the description of the algorithm.

\subsection{Analysis}
Notice that the flow never exceeds one-fourth of the capacity of any edge or node,
since doubling takes only place when the flow is at or below one eighth of capacity.  Since the initial flow is $1/m$ on all paths, it takes at most $O(\log m) = O(\log n)$ rounds before the algorithm is not able to send any more flow using only non-full nodes.  Thus after $O((D+ \sqrt{n}) \log n)$ rounds of broadcast-CONGEST, we have computed flow values which are at least $1/8$ of a \emph{maximal} flow.  We now claim that any maximal flow is close to a maximum flow.  The \emph{depth} of a flow network is the length of the longest $s-t$ path, so in our flow network the depth is $4$.

\begin{lemma} \label{lem:flow-depth}
In any flow network of depth $d$, every maximal flow has value at least $1/d$ of the value of the maximum flow.
\end{lemma}

\begin{proof}
The depth constraint implies that $v(f) \cdot d \ge \sum_{e \in F} f(e)$; namely, since each flow path is of length at most $d$, each unit of flow is counted at most $d$ times in the sum.
Maximality means that there is an $s-t$ cut $(S,V-S)$ such that
all edges in $F$ that go from $S$ to $V-S$ are at full capacity (in $f$).
This implies that $\sum_{v \in S} f(v) \ge cap(S,V-S)$, where $cap(S,V-S)$ is the sum of the edge capacities across the cut. 
Observe that $\sum_{v \in S} f(v) \le \sum_{v \in V \setminus \{t\}} f(v) = \sum_{e \in F} f(e)$.
The capacity constraints imply that $cap(S,V-S) \ge v(f^*)$, 
where $f^*$ is a maximum flow. Combined, we have that $v(f) \ge v(f^*)/d$.
\end{proof}

\begin{corollary} \label{cor:flow-value}
$v(f) \geq \frac{1}{32} |M|$ for every constrained $(1,q)$-matching $M$ of $H$.
\end{corollary}
\begin{proof}
Since in $f$ every $s-t$ path contains at least one full node, $v(f)$ is at least $1/8$ of the the value of any maximal flow, and so by Lemma~\ref{lem:flow-depth} we know that $v(f)$ is at least $1/32$ of the value of the maximum flow.  As discussed, there is a bijection between the integral flows in $F$ and constrained $(1,q)$-matchings in $H$, so this implies that $v(f)$ is at least $1/32$ times the size of the maximum constrained $(1,q)$-matching in $F$. 
\end{proof}

By construction $S'$ is a feasible constrained $(1,q)$-matching, so we just need to show that it has large value.  To do this, we will relate it to $v(F)$

\begin{lemma} \label{lem:S-size}
$\E[|S'|] \ge v(f)/4$.
\end{lemma}

\begin{proof}
Consider an edge $e=(u,v)$ in $H$ (i.e., from $U$ to $Q$).
Let $A_e$ be the event that some other edge incident on $u$ is added to $S$.
Let $B_e$ be the event that $q$ or more edges are added 
that have endpoints in the same bundle as $e$ but not $u$.
Finally, let $C_e$ be the event that $q$ or more other edges incident on $v$ were added to $S$.
Observe that if $e$ was added to $S$, then it will remain in $S'$
if none of the three events ($A_e, B_e$ and $C_e$) take place.
Using that the flow is at most one-fourth of capacity,
\[ \Pr[A_e] = 1 - \prod_{e' \ni u, e' \neq e} (1-f(e')) 
  \le \sum_{e' \ni u} f(e')
  \le 1/4\ . \]
Let $Y$ be the number of edges in $S$ incident on the same bundle as $e$.
Then, $\E[Y] \le q/4$.
So, by Markov's inequality,
\[ \Pr[B_e] \le \Pr[Y \ge 4 \E[Y]] \le 1/4\ . \]
Similarly, $\Pr[C_e] \le 1/4$.
By the union bound,
\[ \Pr[A_e \cup B_e \cup C_e] \le 3/4\ .\]
Thus, the event $X_e$ that edge $e$ is contained in $S'$ has probability 
\[ \Pr[X_e] \ge f(e) \cdot (1 - \Pr[A_e \cup B_e \cup C_e]) = f(e)/4\ , \]
since the event of $e$ being chosen in $S$ is independent from the three bad events.
Thus, by linearity of expectation, 
\[ \E[|S']] = \sum_{e \in E(H)} \Pr[X_e] \ge \frac{1}{4} \sum_{e \in E(H)} f(e) = v(f)/4 \ . \qedhere \] 
\end{proof}

\begin{theorem}
The algorithm finds a $128$-approximate constrained $(1,q)$-matching (in expectation) in time $O((D+\sqrt{n})\log n)$.
\end{theorem}
\begin{proof}
  Lemma~\ref{lem:S-size} and Corollary \ref{cor:flow-value} imply that $\E[|S'|] \ge v(f)/4 \ge OPT/128$, where $OPT$ is the size of an optimal $q$-constrained matching.  So the algorithm returns a 128-approximation to the maximum constrained $(1,q)$-matching.  For the running time, we have already argued that computing $f$ takes at most $O((D+\sqrt{n})\log n)$ rounds.  Computing $S'$ from $f$ clearly takes at most $O(D+\sqrt{n})$ rounds, since it simply involves a constant number of \alg{Component-Broadcast} calls in each leaf branch and bundle.  Thus the total running time is $O((D + \sqrt{n})\log n)$.  
\end{proof}

\fi

\end{document}